\documentclass[twocolumn,showpacs,amsmath,amssymb,aps,pra]{revtex4-1}

\usepackage{graphicx}
\usepackage{dcolumn}
\usepackage{bm}
\usepackage{mathrsfs}
\usepackage{amssymb}
\usepackage{amsmath}
\usepackage{amsthm}

\newtheorem{proposition}{Proposition}

\newtheorem{lemma}{Lemma}
\newtheorem{corollary}{Corollary}

\theoremstyle{definition}

\newcommand{\bra}[1]{\langle #1|}
\newcommand{\ket}[1]{| #1 \rangle }
\newcommand{\ip}[2]{{\langle #1|}{ #2 \rangle }}
\newcommand{\tr}[1]{{\rm tr}[#1]}
\newcommand{\be}{\begin{eqnarray}}
\newcommand{\ee}{\end{eqnarray}}

\newcommand{\cE}{{\cal E}}

\newcommand{\cT}{{\cal T}}

\newcommand{\cS}{{\cal S}}
\newcommand{\cH}{{\cal H}}

\newcommand{\cO}{{\cal O}}

\begin{document}

\title{Bipartite entanglement-annihilating maps: necessary and sufficient conditions}

\author{Sergey N. Filippov$^{1,2,3}$ and M\'{a}rio Ziman$^{1,4}$}
\affiliation{
$^1$Research Center for Quantum Information, Institute of Physics,
Slovak Academy of Sciences, D\'{u}bravsk\'{a} cesta 9, Bratislava 84511, Slovakia \\
$^2$Moscow Institute of Physics and Technology, Institutskii Pereulok 9, Dolgoprudny, Moscow Region 141700, Russia \\
$^3$Institute of Physics and Technology, Russian Academy of Sciences, Nakhimovskii Pr. 34, Moscow 117218, Russia \\
$^4$Faculty of Informatics, Masaryk University, Botanick\'{a} 68a,
Brno 60200, Czech Republic
}

\begin{abstract}
We fully characterize bipartite entanglement-annihilating (EA)
channels that destroy entanglement of any state shared by
subsystems and, thus, should be avoided in any
entanglement-enabled experiment. Our approach relies on extending
the problem to EA positive maps, the cone of which remains
invariant under concatenation with partially positive maps. Due to
this invariancy, positive EA maps adopt a well characterization
and their intersection with completely positive trace-preserving
maps results in the set of EA channels. In addition to a general
description, we also provide sufficient operational criteria
revealing EA channels. They have a clear physical meaning since
the processes involved contain stages of classical information
transfer for subsystems. We demonstrate the applicability of
derived criteria for local and global depolarizing noises, and
specify corresponding noise levels beyond which any initial state
becomes disentangled after passing the channel. The robustness of
some entangled states is discussed.
\end{abstract}

\pacs{03.65.Ud, 03.65.Yz, 03.67.Mn}

\maketitle

\section{Introduction}
Entanglement is a quantum phenomenon with numerous potential
quantum information
applications~\cite{horodecki-2009,plenio-2007}. However, the
practical realization of such applications is typically faced with
various sources of noise, which affect the performance and design
of entanglement-enabled technologies. It is of practical interest
to understand how entanglement is influenced by any such
experimental imperfections. This problem has stimulated
considerable research efforts, which has introduced the concepts
of entanglement sudden death and
revival~\cite{almeida-2007,bellomo-2007,yu-2009}, entanglement
robustness~\cite{simon-2002,bandyopadhyay-2005,man-2008,minert-2009},
entanglement-breaking~\cite{holevo-1998,king-2002,shor-2002,ruskai-2003,horodecki-2003,holevo-2008}
and entanglement-annihilating
processes~\cite{moravcikova-ziman-2010}.

One of the main lessons of entanglement
theory~\cite{horodecki-2009,plenio-2007} is that the presence of
entanglement is in general extremely difficult to verify.
Therefore, some limitations are typically imposed on both initial
states and noise models in most of the studies on the dynamics of
entanglement~\cite{simon-2002,bandyopadhyay-2005,man-2008,minert-2009,zyczkowski-2001,dodd-2004,dur-2004,aolita-2008,frowis-2011}.
No doubt such analysis is in many cases of great practical
relevance, however, the conclusions do not necessarily capture the
universal behavior of entanglement. In fact, it is questionable
whether some universal entanglement dynamics features do exist.
For example, are there processes capable of creating (not
decreasing) entanglement regardless of the initial state? Or, on
the other hand, are there processes that destroy any entanglement?
Is there some equation capturing the dynamics of entanglement?

The first of these questions resulted in considering various
aspects of entangling and disentangling capabilities of quantum
processes \cite{zanardi-2000,linden-2009}. The other two questions
were mostly studied for one-side noisy processes $\Phi\otimes{\rm
Id}$, where the noise $\Phi$ acts only on one of the subsystems
while the rest of subsystems evolve in a noiseless manner (${\rm
Id}$). For such processes, the so-called evolution equation for
entanglement has been
derived~\cite{konrad-2008,tiersch-2008,gour-2010,zhang-2010}. It
says that the change of the entanglement due to one-sided noise is
quantitatively bounded by its action on the maximally entangled
state. Then, all the noises $\Phi$ that disentangle the maximally
entangled state will also disentangle a given subsystem (under the
noise action) from any other subsystem (noiseless) regardless of
the initial state of the global system, {\it ipso facto} forming a
class of entanglement-breaking (EB)
processes~\cite{holevo-1998,king-2002,shor-2002,ruskai-2003,horodecki-2003,holevo-2008}.

In practice, however, the noise is rarely one-sided. This is the
reason why the notion of entanglement-annihilating (EA) processes
was introduced in Ref.~\cite{moravcikova-ziman-2010}. Formally,
the noise (not necessarily one-sided, or local) is EA if its
action disentangles all the subsystems forming the composite
system. EA processes acting on a composite system do not
necessarily disentangle the system from its surrounding, they only
destroy entanglement between subsystems accessible in experiment.
For instance, it could happen that the joint action of local
noises on individual subsystems constitutes an EA process even if
none of the local noises is
EB~\cite{moravcikova-ziman-2010,filippov-rybar-ziman-2012}.

Although EA processes impose fundamental limitations on the
performance of entanglement-enabled experiments, they are not
explored much. In this paper, we provide explicit characterization
of general bipartite EA channels and derive sufficient criteria
for their detection. We employ these criteria to specify the
maximal noise levels above which no entanglement can be preserved.

\section{Preliminaries}
The states of a quantum system associated with a $d$-dimensional
Hilbert space $\cH_d$ are identified with density operators
(positive and unit trace) and form a convex set $\cS(\cH_d)$.
Quantum processes are modelled as channels, i.e. completely
positive trace-preserving (CPT) linear maps $\Phi:\cT(\cH_{\rm
in})\to\cT(\cH_{\rm out})$ on trace-class operators $\cT(\cH_{\rm
in})$. We say a state of the system $S$ composed of subsystems
$\mathscr{A,B,\dots}$ is separable if $\varrho= \sum_j p_j
\varrho^{\mathscr{A}}_j\otimes\varrho^{\mathscr{B}}_j\otimes\cdots$,
with $\{p_j\}$ being a probability distribution. Otherwise it is
called entangled. We say a channel $\Phi^S \equiv
\Phi^{\mathscr{AB\cdots}}$ is EA if $\Phi^S[\varrho^S]$ is
separable (w.r.t. partition $\mathscr{A|B|C}|\ldots$) for all
input states. $\Phi^S$ is EB if $(\Phi^S\otimes{\rm
Id}^E)[\omega^{SE}]$ is separable w.r.t. partition $S|E$ for all
states $\omega^{SE}$ of system $S$ and an arbitrary environment
$E$.

Quantum channel $\Phi$ can be written in a (non-unique) sum
diagonal representation $\Phi[\varrho] = \sum_k A_k \varrho
A_k^{\dag}$, where Kraus operators $A_k$  satisfy the
normalization $\sum_k A_k^{\dag}A_k = I_{\rm in}$ (identity
operator). Because of that, the channel $\Phi$ can be seen as a
sum of conditional outputs of a measurement, in which the outcomes
$k$ are occurring with probability $p_k=\tr{\varrho
A_k^{\dag}A_k}$ while the state is undergoing the conditional
(post-selected) transformation $\varrho\mapsto p_k^{-1}A_k \varrho
A_k^\dagger$ [Figs.~\ref{figure1}(a) and \ref{figure1}(b)].

Any linear map $\Phi:\cT(\cH_{\rm in}^S)\rightarrow\cT(\cH_{\rm
out}^S)$ can be described by a so-called Choi
matrix~\cite{choi-1975,jamiolkowski-1972}
\begin{equation}
\label{choi-matrix} \Omega_{\Phi}^{SS'} := (\Phi^S \otimes {\rm
Id}^{S'})[\ket{\Psi_+^{SS'}}\bra{\Psi_+^{SS'}}],
\end{equation}
where $\ket{\Psi_+^{SS'}} = (d^S)^{-1/2} \sum_{i=1}^{d^S}
\ket{i\otimes i'}$ is a maximally entangled state shared by system
$S$ and its clone $S'$, $\ip{i}{j}=\ip{i'}{j'}=\delta_{ij}$. It is
well known~\cite{choi-1975,jamiolkowski-1972} that the map
$\Phi^S$ is completely positive (CP) if and only if
$\Omega_{\Phi}^{SS'} \ge 0$, i.e.
$\Omega_{\Phi}^{SS'}\in\cS(\cH_{\rm out}^{S}\otimes\cH_{\rm
in}^{S'})$. The matrix in Eq.~(\ref{choi-matrix}) defines the map:
\begin{equation}
\label{map-through-choi} \Phi [X] = d^S \, {\rm tr}_{S'} \, [
\,\Omega_{\Phi}^{SS'} (I_{\rm out}^S \otimes X^{\rm T}) \, ],
\end{equation}
where $X^{\rm T}=\sum_{i,j} \bra{j} X \ket{i} \ket{i'}\bra{j'} \in
\cT(\cH_{\rm in}^{S'})$ and ${\rm tr}_{S'}$ denotes the partial
trace operation.

%%%%%%%%%%%%%%%%%%%%%%%%%%%%%%%%%%%%%%%%%%%%%%%%%%%%%%%%%%%%%%%%%%%
%%%%%%%%%%%%%%%%%%%%%%%%%%%%%%%%%%%%%%%%%%%%%%%%%%%%%%%%%%%%%%%%%%%
\begin{figure}
\includegraphics[width=8.5cm]{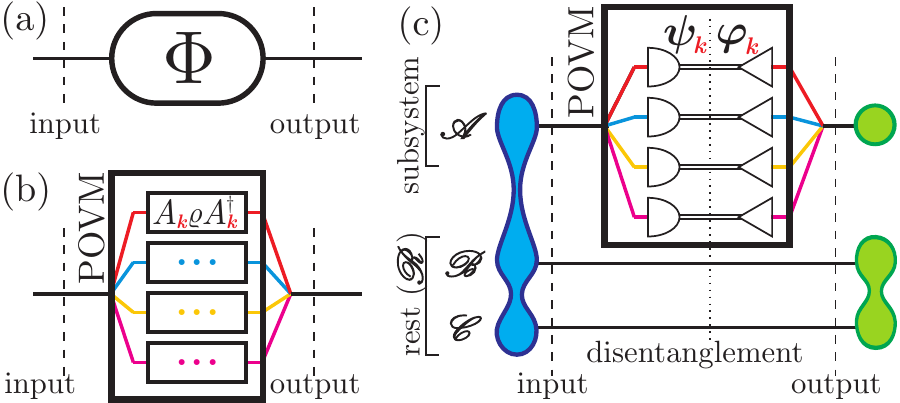}
\caption{\label{figure1} (a) Quantum channel as an input--output
device. (b) Physical interpretation of the diagonal-sum
representation. (c) Structure of EB channels.}
\end{figure}
%%%%%%%%%%%%%%%%%%%%%%%%%%%%%%%%%%%%%%%%%%%%%%%%%%%%%%%%%%%%%%%%%%%
%%%%%%%%%%%%%%%%%%%%%%%%%%%%%%%%%%%%%%%%%%%%%%%%%%%%%%%%%%%%%%%%%%%

A general positive linear map $\Lambda$ that transforms positive
operators into positive ones gives rise to the Choi matrix which
can be non-positive in general. We say an operator
$\xi\in\cT(\cH^{\mathscr{X}}\otimes\cH^{\mathscr{Y}})$ is
block-positive  (denoted as $\xi_{\rm BP}^{\mathscr{X|Y}}$) if
$\bra{x\otimes y} \xi \ket{x \otimes y} \ge 0$ for all
$\ket{x}\in\cH^{\mathscr{X}}$, $\ket{y}\in\cH^{\mathscr{Y}}$.
Then, $\{\Lambda^S$ is positive$\} \Leftrightarrow
\{\Omega_{\Lambda}^{SS'}$ is
block-positive$\}$~\cite{jamiolkowski-1972}.

\subsection*{EB channels} Suppose subsystem $\mathscr{A}$ is
subjected to a quantum channel $\Phi^{\mathscr{A}}$ whose Kraus
operators are rank-1 projectors, i.e. $A_k \propto \ket{\varphi_k}
\bra{\psi_k}$ with $\ket{\psi_k}\in\cH_{\rm in}^{\mathscr{A}}$ and
$\ket{\varphi_k}\in\cH_{\rm out}^{\mathscr{A}}$. In this case, we
deal with a measure-and-prepare procedure, i.e. the channel of
Holevo form~\cite{holevo-1998,horodecki-2003,holevo-2008}. Such
channel $\Phi^{\mathscr{A}}$ is EB and disentangles $\mathscr{A}$
from all rest subsystems $\mathscr{B,C,\ldots(=R)}$ because
contains a stage of classical information transfer (depicted by
double line in Fig.~\ref{figure1}c). Surprisingly, the converse
statement is also true, i.e. $\{\Phi$ is EB$\} \Leftrightarrow
\{$there exists a diagonal sum representation of $\Phi$ with
rank-1 Kraus operators$\}$. Alternative characterization of EB
channels exploits the property of Choi matrix:
$\{\Phi^{\mathscr{A}}$ is EB$\} \Leftrightarrow
\{\Omega_{\Phi}^{\mathscr{AA'}}\in\cS(\cH_{\rm
out}^{\mathscr{A}}\otimes\cH_{\rm in}^{\mathscr{A'}})$ is
separable w.r.t. partition
$\mathscr{A}|\mathscr{A'}\}$~\cite{ruskai-2003,horodecki-2003,holevo-2008}.

As far as EB channels $\Phi_{\rm EB}^{\mathscr{A\!B}}$ acting on a
composite system $\mathscr{A\!B}$ are concerned, the Choi state
$\Omega_{\Phi}^{\mathscr{A\!BA'\!B'}}$ is to be separable w.r.t.
partition $\mathscr{A\!B|A'\!B'}$ but can still be entangled
w.r.t. partitions $\mathscr{A|BA'\!B'}$ and $\mathscr{B|AA'\!B'}$
(for instance, if $\Omega_{\Phi}^{\mathscr{A\!BA'\!B'}}$ is the
Smolin state of 4 qubits~\cite{smolin-2001}). In this case, the
channel disentangles ($\mathscr{AB}$) from any other systems
$\mathscr{C,D},\ldots$ but the entanglement between $\mathscr{A}$
and $\mathscr{B}$ can be preserved. However, if the channel
$\Phi^{\mathscr{A\!B}}$ has a local structure $\Phi_{\rm
local}^{\mathscr{A\!B}} = \Phi_1^{\mathscr{A}} \otimes
\Phi_2^{\mathscr{B}}$, then $\{\Phi_{\rm local}^{\mathscr{A\!B}}$
is EB$\} \Leftrightarrow \{\Phi_1^{\mathscr{A}}$ is EB and
$\Phi_2^{\mathscr{B}}$ is EB$\}$, which follows immediately from
the particular form of the maximally entangled state
$\ket{\Psi_+^{\mathscr{A\!B|A'\!B'}}} := \ket{\Psi_+^{\mathscr{AA'}}} \otimes
\ket{\Psi_+^{\mathscr{BB'}}}$.

\subsection*{EA channels}
In contrast to EB channels, EA channels by definition act on
composite systems. For bipartite systems one can use the Horodecki
criterion~\cite{horodecki-1996} to formulate a necessary and
sufficient condition for the map to be EA.
\begin{lemma}
\label{lemma-EA} Suppose $\Phi^{\mathscr{A\!B}}:\cT(\cH_{\rm
in}^{\mathscr{A\!B}})\rightarrow\cT(\cH_{\rm
out}^{\mathscr{A\!B}})$ is a channel. Then
$\{\Phi^{\mathscr{A\!B}}$ is EA$\}\Leftrightarrow\{({\rm
Id}^{\mathscr{A}}\otimes\Lambda^{\mathscr{B}})\circ\Phi^{\mathscr{A\!B}}$
is a positive map for any positive map
$\Lambda^{\mathscr{B}}:\cT(\cH_{\rm
out}^{\mathscr{B}})\rightarrow\cT(\cH_{\rm out}^{\mathscr{A}})\}$.
\end{lemma}
\noindent Unfortunately, Lemma~\ref{lemma-EA} is not quite
operational, which makes it difficult to apply. However, in the
case of two qubits ($d_{\rm out}^{\mathscr{A,B}} = 2$),
Lemma~\ref{lemma-EA} turns out to be rather fruitful because
without loss of generality the positive map
$\Lambda^{\mathscr{B}}$ can be chosen to be either a
transposition~\cite{peres-1996} or a reduction
map~\cite{horodecki-reduction-1999}. This fact was exploited in
characterization of local two-qubit EA channels in
Ref.~\cite{filippov-rybar-ziman-2012}. Particularly interesting in
the case of bipartite local channels $\Phi^{\mathscr{A\!B}}$ with
$d^{\mathscr{A}}=d^{\mathscr{B}}$ are those that form
$\Phi\otimes\Phi$ describing the physical situations when both
parties experience the same noise.
Following~\cite{moravcikova-ziman-2010}, if $\Phi\otimes\Phi$ is
EA, we will refer to a ``generating'' channel $\Phi$ as a
2-locally EA channel (2LEA).

\subsection*{Structure of linear bipartite maps}
To investigate the structure of EA channels it turns out to be
instructive to introduce the concept of positive
entanglement-annihilating (PEA) linear maps. In particular, a map
$\Phi^{\mathscr{A\!B}}$ is PEA if it is positive and
$\Phi^{\mathscr{A\!B}}[\varrho]$ belongs to a cone of states
separable w.r.t. partition $\mathscr{A|B}$ for all
$\varrho\in\cS(\cH_{\rm in}^{\mathscr{A\!B}})$. The set of PEA
maps is convex and its intersection with CPT maps gives exactly
all EA channels, i.e. EA = PEA $\cap$ CPT.

Consider an example of 2-locally unital qubit linear
trace-preserving maps, i.e. maps of the form
$\Upsilon\otimes\Upsilon$ with $\Upsilon[I]=I$. Up to a unitary
preprocessing and postprocessing the map $\Upsilon$ can be
written~\cite{ruskai-2002} in the form $\Upsilon [X] = \frac{1}{2}
\sum_{j=0}^3 \lambda_j \tr{\sigma_j X} \sigma_j$, where
$\{\lambda_j\}$ are real numbers, $\sigma_0 = I$, and
$\{\sigma_i\}_{j=1}^3$ is a conventional set of Pauli operators
(in an appropriate basis). Due to the trace-preserving condition,
$\lambda_0= 1$. The remaining three parameters
$\{\lambda_j\}_{j=1}^3$ are scaling coefficients of Bloch ball
axes. The map $\Upsilon$ is given by a point in the Cartesian
coordinate system $(\lambda_1,\lambda_2,\lambda_3)$ and the
following relations hold:
\newline(\!{\it i}) $\{\Upsilon$ is
positive$\}\Leftrightarrow\{|\lambda_j|\le 1, j=1,2,3\}$;
\newline(\!{\it ii}) $\{\Upsilon$ is CP$\}\Leftrightarrow\{\Upsilon\otimes\Upsilon$
is CP$\}\Leftrightarrow 1\pm\lambda_3\geq |\lambda_1\pm
\lambda_2|$;
\newline(\!{\it iii}) $\{\Upsilon$ is EB$\}\Leftrightarrow\{\Upsilon\otimes\Upsilon$ is EB$\}\Leftrightarrow\{|\lambda_1|+|\lambda_2|+|\lambda_3| \le 1\}$;
\newline(\!{\it iv}) $\{\Upsilon\otimes\Upsilon$ is
positive$\}\Leftrightarrow\{\Upsilon^2$ is CP$\}\Leftrightarrow
1\pm\lambda_3^2\geq |\lambda_1^2\pm \lambda_2^2|$;
\newline(\!{\it v}) $\{\Upsilon\otimes\Upsilon$ is PEA$\}\Leftrightarrow\{\Upsilon^2$ is EB$\}\Leftrightarrow\{\lambda_1^2+\lambda_2^2+\lambda_3^2\le 1\}$;
\newline(\!{\it vi}) $\{\Upsilon\otimes\Upsilon$ is EA$\}\Leftrightarrow\{\Upsilon$ is CP and $\Upsilon^2$ is
EB$\}$.

Clearly (\!{\it vi}) = (\!{\it ii}) $\cap$ (\!{\it v}). These are
analogies [(\!{\it ii})$\sim$(\!{\it iii})] and [(\!{\it
iv})$\sim$(\!{\it v})] that stimulated us to extend the concept of
entanglement annihilation to positive maps. Indeed, definitions of
both the CP and EB maps require extensions of the channel action,
whereas the concepts of positive and PEA maps do not require for
their definition any additional physical system. The structure of
linear bipartite maps is illustrated in Fig.~\ref{figure2}.

\section{Criteria} The appealing simplicity of item (\!{\it v}) above is not
sudden and holds due to a general property that the cone of PEA
maps is closed under left-composition by partially positive maps
${\rm Id}^{\mathscr{A}}\otimes\Lambda^{\mathscr{B}}$ (left
PP-invariant), which follows from Lemma~\ref{lemma-EA}. This
fundamental property enables us to characterize PEA maps.
\begin{proposition}
\label{prop-1} The map $\Phi^{\mathscr{A\!B}}$ is PEA if and only
if
\begin{equation}
\label{proposition-1} {\rm tr} \left[ (\xi_{\rm
BP}^{\mathscr{A|B}}\otimes\varrho^{\mathscr{A'\!B'}})
\Omega_{\Phi}^{\mathscr{A\!BA'\!B'}} \right] \ge 0
\end{equation}
\noindent for all block-positive $\xi_{\rm BP}^{\mathscr{A|B}}$
and $\varrho^{\mathscr{A'\!B'}} \in
\cS(\cH^{\mathscr{A'}}\!\otimes\cH^{\mathscr{B'}})$.
\end{proposition}
\begin{proof}
Using the extension of Lemma~\ref{lemma-EA} for positive maps, we
get $\{\Phi^{\mathscr{A\!B}}$ is PEA$\}\Leftrightarrow \{({\rm
Id}^{\mathscr{A}}\otimes\Lambda^{\mathscr{B}})\circ\Phi^{\mathscr{A\!B}}$
is a positive map for any positive map $\Lambda^{\mathscr{B}}\}$,
which is equivalent to the block-positivity of matrix
$\Omega_{({\rm Id}\otimes\Lambda)\circ\Phi}^{\mathscr{A\!BA'\!B'}}
\equiv ({\rm
Id}^{\mathscr{AA'\!B'}}\otimes\Lambda^{\mathscr{B}})[\Omega_{\Phi}^{\mathscr{A\!BA'\!B'}}]$.
By definition of block-positivity,
\begin{eqnarray}
&& {\rm tr} \Big\{ \ket{\varphi^{\mathscr{A\!B}} \otimes
\chi^{\mathscr{A'\!B'}}} \bra{\varphi^{\mathscr{A\!B}} \otimes
\chi^{\mathscr{A'\!B'}}} \nonumber\\
&& \qquad\qquad\qquad\qquad\quad \times ({\rm
Id}^{\mathscr{A\!A'\!B'}} \otimes \Lambda^{\mathscr{B}})
[\Omega_{\Phi}^{\mathscr{A\!BA'\!B'}}] \Big\} \nonumber\\
&& \equiv {\rm tr} \Big\{ ({\rm Id}^{\mathscr{A}} \otimes
\Lambda^{\dag\mathscr{B}})
[\ket{\varphi^{\mathscr{A\!B}}}\bra{\varphi^{\mathscr{A\!B}}}]
\nonumber\\
&& \qquad\qquad\qquad\qquad\quad \otimes
\ket{\chi^{\mathscr{A'\!B'}}}\bra{\chi^{\mathscr{A'\!B'}}} \
\Omega_{\Phi}^{\mathscr{A\!BA'\!B'}} \Big\} \ge 0, \nonumber
\end{eqnarray}

\noindent where $\Lambda^{\dag}$ denotes the dual map:
$\tr{X\Lambda[Y]} \equiv \tr{\Lambda^{\dag}[X] Y}$. Since the dual
of a positive map is also positive (see,
e.g.,~\cite{johnston-2011}), $\Lambda^{\dag}$ is a positive map
and the operator $({\rm
Id}^{\mathscr{A}}\!\otimes\!\Lambda^{\dag\mathscr{B}})
[\ket{\varphi^{\mathscr{A\!B}}}\bra{\varphi^{\mathscr{A\!B}}}]$ is
block-positive (equals $\xi_{\rm BP}^{\mathscr{A|B}}$). Taking
into account arbitrariness of $\Lambda,\ket{\varphi},\ket{\chi}$
and remembering the convex structure of density operators, we
obtain formula (\ref{proposition-1}).
\end{proof}

Proposition~\ref{prop-1} says (in terms of Choi matrices) that the
cone of PEA maps  is dual to the cone of operators $\xi_{\rm
BP}^{\mathscr{A|B}}\otimes\varrho^{\mathscr{A'\!B'}}$ inducing
[via formula~(\ref{map-through-choi})] maps of the form $\Phi_{\rm
d.c.}^{\mathscr{A\!B}}[X] = \sum_{k} \tr{F_k X} \xi_{{\rm
BP}~k}^{\mathscr{A|B}}$, $F_k \ge 0$. Moreover, using ${\rm tr} =
{\rm tr}_{\mathscr{A\!B}} \circ {\rm tr}_{\mathscr{A'\!B'}}$, we
obtain alternative forms of the condition in
Eq.~(\ref{proposition-1}). In particular, $\Phi^{\mathscr{A\!B}}$
is PEA if and only if for all $\xi_{\rm BP}^{\mathscr{A|B}}$ the
operator ${\rm tr}_{\mathscr{A\!B}} [ (\xi_{\rm
BP}^{\mathscr{A|B}}\otimes I^{\mathscr{A'\!B'}})
\Omega_{\Phi}^{\mathscr{A\!BA'\!B'}} ]$ is positive, i.e. belongs
to ${\rm Cone}(\cS(\cH^{\mathscr{A'\!B'}}))$, or equivalently, if
the operator $\bra{\chi^{\mathscr{A'\!B'}}}
\Omega_{\Phi}^{\mathscr{A\!BA'\!B'}}
\ket{\chi^{\mathscr{A'\!B'}}}$ belongs to a cone of separable
states (w.r.t. partition $\mathscr{A|B}$) for all
$\ket{\chi^{\mathscr{A'\!B'}}}$.

%%%%%%%%%%%%%%%%%%%%%%%%%%%%%%%%%%%%%%%%%%%%%%%%%%%%%%%%%%%%%%%%%%%
%%%%%%%%%%%%%%%%%%%%%%%%%%%%%%%%%%%%%%%%%%%%%%%%%%%%%%%%%%%%%%%%%%%
\begin{figure}
\includegraphics[width=8.5cm]{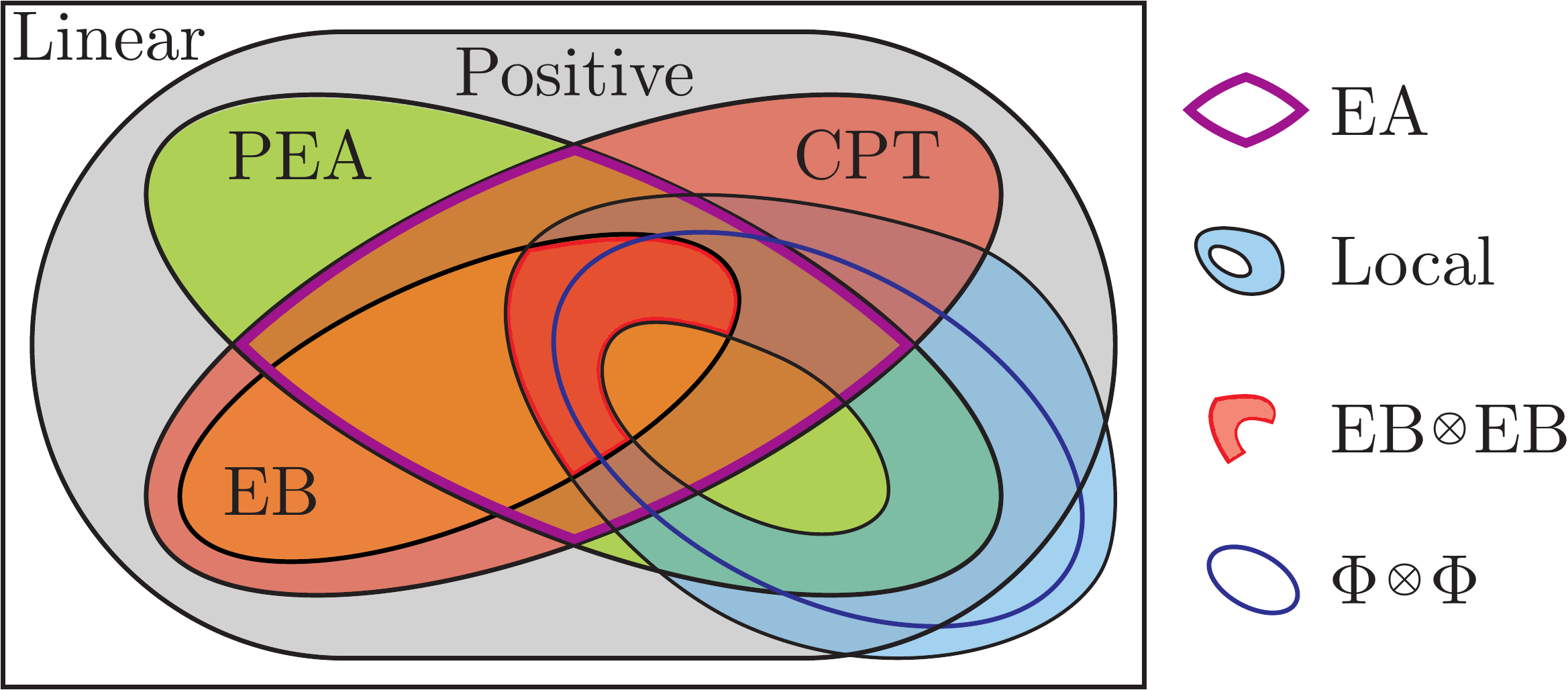}
\caption{\label{figure2} Venn diagram of linear bipartite maps
$\Phi^{\mathscr{A\!B}}$. Convex figures correspond to convex
sets.}
\end{figure}
%%%%%%%%%%%%%%%%%%%%%%%%%%%%%%%%%%%%%%%%%%%%%%%%%%%%%%%%%%%%%%%%%%%
%%%%%%%%%%%%%%%%%%%%%%%%%%%%%%%%%%%%%%%%%%%%%%%%%%%%%%%%%%%%%%%%%%%

We already know that EA = PEA $\cap$ CPT, therefore the complete
characterization of EA channels is as follows.
\begin{corollary}
The linear map $\Phi^{\mathscr{A\!B}}$ is an EA channel if and
only if its Choi matrix $\Omega_{\Phi}^{\mathscr{A\!BA'\!B'}}$
satisfies (\ref{proposition-1}),
$\Omega_{\Phi}^{\mathscr{A\!BA'\!B'}} \ge 0$, and ${\rm
tr}_{\mathscr{A\!B}}\Omega_{\Phi}^{\mathscr{A\!BA'\!B'}} =
(d^{\mathscr{A}}d^{\mathscr{B}})^{-1} I^{\mathscr{A'\!B'}}$.
\end{corollary}
\begin{proof}
The three requirements guarantee that $\Phi\in$ PEA, $\Phi\in$ CP,
and $\Phi$ is trace-preserving, respectively.
\end{proof}

Although Proposition~\ref{prop-1} provides the necessary and
sufficient condition for the map to be PEA, it is challenging to
apply it to a given map. The following proposition provides a
non-trivial sufficient condition which is quite useful as we
demonstrate later.

\begin{proposition}
\label{prop-2} If $\Omega_{\Phi}^{\mathscr{A\!BA'\!B'}}$ can be
written as a convex sum of operators $\zeta_{\rm
BP}^{\mathscr{A|A'\!B'}} \otimes \varrho^{\mathscr{B}}$ and
$\varrho^{\mathscr{A}} \otimes \zeta_{\rm
BP}^{\mathscr{B|A'\!B'}}$, where $\zeta_{\rm BP}$ is
block-positive w.r.t. corresponding cut and $\varrho$ is positive,
then the map $\Phi^{\mathscr{A\!B}}$ is PEA.
\end{proposition}
\begin{proof}
Substituting $\zeta_{\rm BP}^{\mathscr{A|A'\!B'}} \otimes
\varrho^{\mathscr{B}}$ for $\Omega_{\Phi}^{\mathscr{A\!BA'\!B'}}$
in (\ref{proposition-1}), we obtain that ${\rm
tr}_{\mathscr{A'\!B'}} [\zeta_{\rm BP}^{\mathscr{A|A'\!B'}}
\ket{\chi^{\mathscr{A'\!B'}}\!} \bra{\chi^{\mathscr{A'\!B'}}\!} ]
= \tilde{\varrho}^{\mathscr{A}} \ge 0$ and ${\rm
tr}_{\mathscr{A\!B}} [ \xi_{\rm BP}^{\mathscr{A|B}}
\tilde{\varrho}^{\mathscr{A}} \otimes \varrho^{\mathscr{B}} ] \ge
0$, thus, Eq.~(\ref{proposition-1}) holds. By exchanging
$\mathscr{A}\leftrightarrow\mathscr{B}$ it is clear that the
operator $\varrho^{\mathscr{A}} \otimes \xi_{\rm
BP}^{\mathscr{B|A'\!B'}}$ also satisfies the requirement
(\ref{proposition-1}).
\end{proof}

Define $A^{\mathscr{B}}
=\ket{\varphi^{\mathscr{B}}}\bra{\psi^{\mathscr{B}}}$, then
$\zeta_{\rm BP}^{\mathscr{A|A'\!B'}} \otimes
\ket{\varphi^{\mathscr{B}}}\bra{\varphi^{\mathscr{B}}} =
(I^{\mathscr{A\!A'\!B'}}\otimes A^{\mathscr{B}}) \Xi_{{\rm
BP}}^{\mathscr{A\!B|A'\!B'}} (I^{\mathscr{A\!A'\!B'}}\otimes
A^{\mathscr{B}\dag})$ for a suitable $\Xi_{{\rm
BP}}^{\mathscr{A\!B|A'\!B'}}$. Consequently, the map corresponding
to $\zeta_{\rm BP}^{\mathscr{A|A'\!B'}} \otimes
\ket{\varphi^{\mathscr{B}}}\bra{\varphi^{\mathscr{B}}}$ is a
concatenation of a positive map $\Lambda^{\mathscr{A\!B}}$ (given
by Choi matrix $\Xi_{{\rm BP}}^{\mathscr{A\!B|A'\!B'}}$) followed
by an EB operation $\cO_{{\rm EB}}[\bullet] = A \bullet A^{\dag}$
acting on subsystem $\mathscr{B}$. Similarly,
$\varrho^{\mathscr{A}} \otimes \zeta_{\rm
BP}^{\mathscr{B|A'\!B'}}$ describes a positive map on
$\mathscr{A\!B}$ followed by some EB operation applied to
subsystem $\mathscr{A}$. As a result, the subset of PEA maps
characterized by Proposition~\ref{prop-2} can be understood as
mixture of concatenations of positive maps with EB operations
applied on one of the subsystems (see also Fig.~\ref{figure3})
\begin{equation}
\label{EA-resolution}
\Phi^{\mathscr{A\!B}} \!=\! \sum_{k} (\cO_{{\rm EB}~k}^{\mathscr{A(B)}}\otimes {\rm Id}^{\mathscr{B(A)}})\circ \Lambda_k^{\mathscr{A\!B}}\,.
\end{equation}

%%%%%%%%%%%%%%%%%%%%%%%%%%%%%%%%%%%%%%%%%%%%%%%%%%%%%%%%%%%%%%%%%%%
%%%%%%%%%%%%%%%%%%%%%%%%%%%%%%%%%%%%%%%%%%%%%%%%%%%%%%%%%%%%%%%%%%%
\begin{figure}
\includegraphics[width=8.5cm]{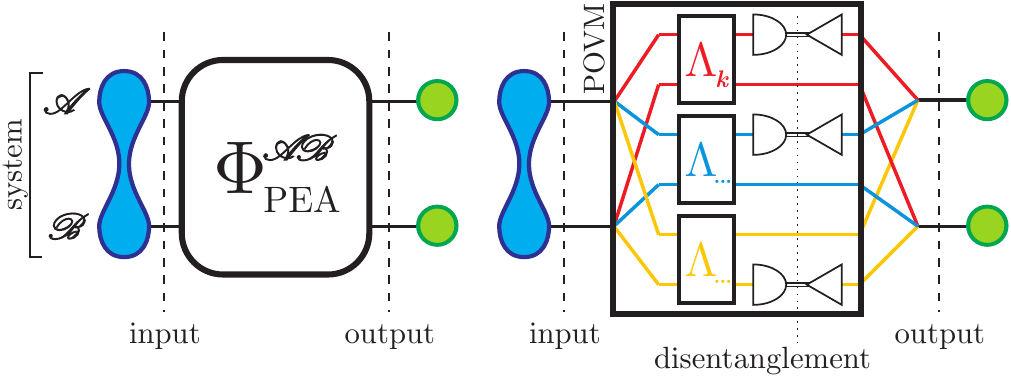}
\caption{\label{figure3} PEA maps and physical meaning of
Proposition~\ref{prop-2} (positive maps $\{\Lambda_k\}$ are
followed by one-sided EB operations).}
\end{figure}
%%%%%%%%%%%%%%%%%%%%%%%%%%%%%%%%%%%%%%%%%%%%%%%%%%%%%%%%%%%%%%%%%%%
%%%%%%%%%%%%%%%%%%%%%%%%%%%%%%%%%%%%%%%%%%%%%%%%%%%%%%%%%%%%%%%%%%%

Moreover, if we replace in Proposition~\ref{prop-2} the
block-positive operators $\zeta_{\rm BP}^{\mathscr{A|A'\!B'}}$ and
$\zeta_{\rm BP}^{\mathscr{B|A'\!B'}}$ by positive ones
$\varrho^{\mathscr{AA'\!B'}}$ and $\varrho^{\mathscr{BA'\!B'}}$,
respectively, then the corresponding Choi matrix will
automatically be positive and the associated map will be a fair CP
map.
\begin{corollary}
\label{corollary-biseparable} If ${\rm
tr}_{\mathscr{A\!B}}\Omega_{\Phi}^{\mathscr{A\!BA'\!B'}} =
(d^{\mathscr{A}}d^{\mathscr{B}})^{-1} I^{\mathscr{A'\!B'}}$ and
$\Omega_{\Phi}^{\mathscr{A\!BA'\!B'}}$ is a convex sum of density
operators $\varrho^{\mathscr{A|BA'\!B'}}$ and
$\varrho^{\mathscr{B|AA'\!B'}}$ (separable w.r.t. partitions
$\mathscr{A|BA'\!B'}$ and $\mathscr{B|AA'\!B'}$, respectively),
then $\Phi^{\mathscr{A\!B}}$ is an EA channel.
\end{corollary}

Let us note that such states
$\Omega_{\Phi}^{\mathscr{A\!BA'\!B'}}$ belong to a family of
so-called biseparable states (convex hull of states separable with
respect to some bipartite cut). Unfortunately, only a little is
known about biseparability
detection~\cite{huber-2010,barreiro-2010,jungnitsch-2011,kampermann-2012},
however, Corollary~\ref{corollary-biseparable} encourages its
deeper investigation (see, e.g., a recent approach in
Ref.~\cite{huber-2013}).

\section{Case study: depolarizing channels} Given a quantum
channel $\Phi^{\mathscr{A\!B}}$, one can settle the question of
its being EA in the affirmative by finding either the
resolution~(\ref{EA-resolution}) or the resolution of
Corollary~\ref{corollary-biseparable}. Once the resolution is
found, it guarantees that $\Phi^{\mathscr{A\!B}}$ is a PEA map,
and, consequently, the channel is EA. As an example we examine a
family of depolarizing channels which can act either locally or
globally on the system reflecting the physical situation of
individual or common baths, respectively.

%%%%%%%%%%%%%%%%%%%%%%%%%%%%%%%%%%%%%%%%%%%%%%%%%%%%%%%%%%%%%%%%%%%
%%%%%%%%%%%%%%%%%%%%%%%%%%%%%%%%%%%%%%%%%%%%%%%%%%%%%%%%%%%%%%%%%%%
\begin{figure}
\includegraphics[width=8.5cm]{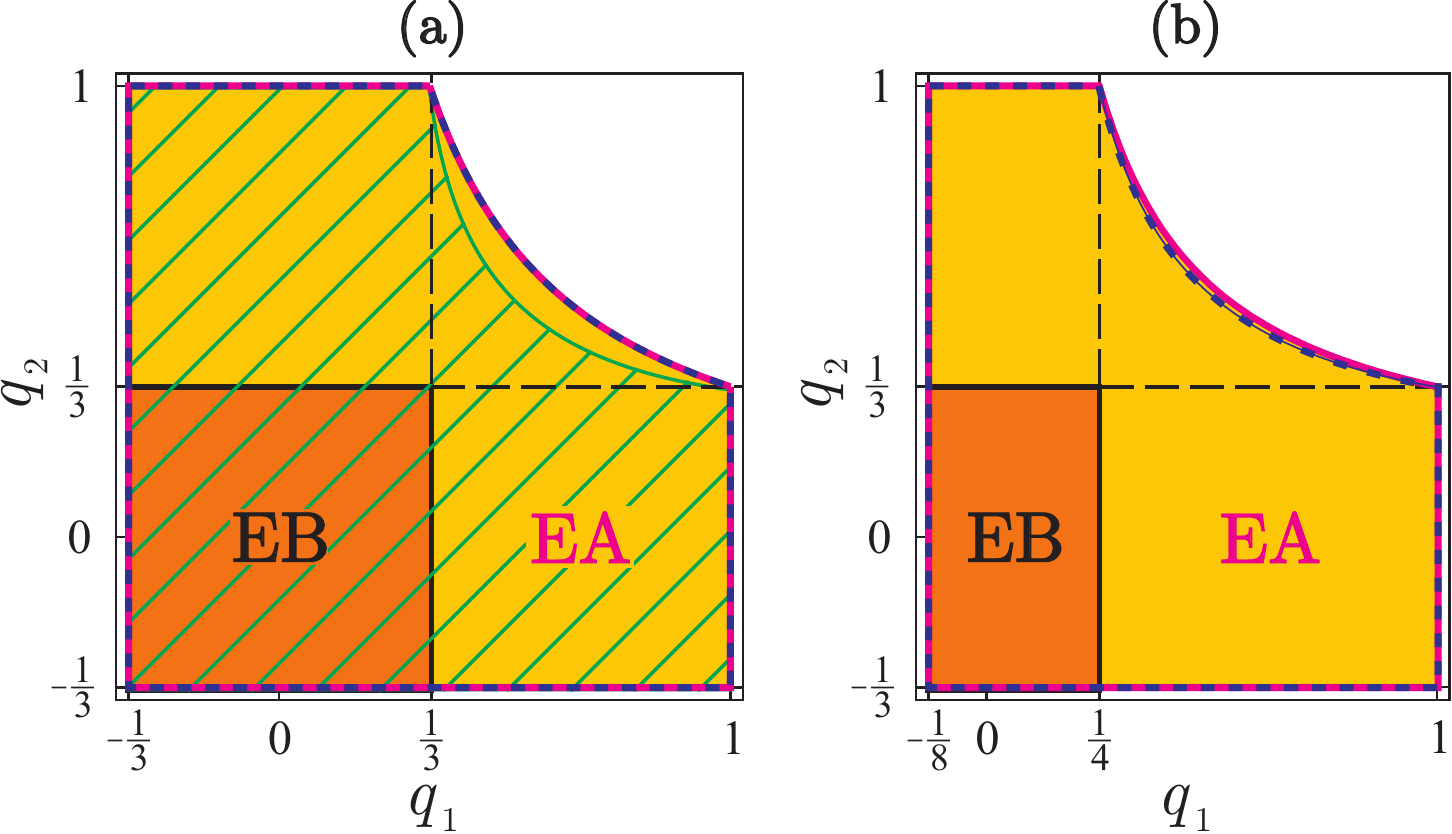}
\caption{\label{figure4} Area of parameters $(q_1,q_2)$ where the
local-depolarizing channel
$\Phi_{q_1}^{\mathscr{A}}\otimes\Phi_{q_2}^{\mathscr{B}}$ is EA:
(a) two-qubit system, $d^{\mathscr{A}}=d^{\mathscr{B}}=2$; (b)
qutrit-qubit system, $d^{\mathscr{A}}=3$, $d^{\mathscr{B}}=2$. The
validity of Proposition~\ref{prop-2} is justified inside the
dashed region, which provides all EA channels in case (a) and a
subset of EA channels in case (b). Being applied to case (a),
Corollary~\ref{corollary-biseparable} detects EA behavior inside
the green hatching.}
\end{figure}
%%%%%%%%%%%%%%%%%%%%%%%%%%%%%%%%%%%%%%%%%%%%%%%%%%%%%%%%%%%%%%%%%%%
%%%%%%%%%%%%%%%%%%%%%%%%%%%%%%%%%%%%%%%%%%%%%%%%%%%%%%%%%%%%%%%%%%%

The depolarizing channel on $d$-dimensional system is defined
through $\Phi_q = q {\rm Id} + (1-q) {\rm Tr}$, where ${\rm Tr}[X]
= \tr{X} \frac{1}{\rm d} I_d$ is the trace map and
$q\in[-\frac{1}{d^2-1},1]$. Note that $\Phi_q$ is EB if and only
if $-\frac{1}{d^2-1} \le q \le \frac{1}{d+1}$ (see
Appendix~\ref{appendix-b}). A bipartite system $\mathscr{A\!B}$
can be affected by a local depolarizing noise of the form
$\Phi_{q_1}^{\mathscr{A}}\otimes\Phi_{q_2}^{\mathscr{B}}$, or a
global depolarizing noise of the form $\Phi_{q}^{\mathscr{A\!B}}$.

Firstly, we illustrate the efficiency of the derived criteria by
examples of $2\times 2$ and $3\times 2$ systems for which the
exact solutions can be readily found thanks to the
Peres--Horodecki criterion~\cite{peres-1996,horodecki-1996}: in
the case $d^{\mathscr{A}}=d^{\mathscr{B}}=2$,
$\Phi_{q_1}^{\mathscr{A}}\otimes\Phi_{q_2}^{\mathscr{B}}$ is EA if
$q_1 q_2 \le \frac{1}{3}$ and $\Phi_{q}^{\mathscr{A\!B}}$ is EA if
$q \le \frac{1}{3}$; in the case $d^{\mathscr{A}}=3$ and
$d^{\mathscr{B}}=2$,
$\Phi_{q_1}^{\mathscr{A}}\otimes\Phi_{q_2}^{\mathscr{B}}$ is EA if
$q_1 (9 q_2 - 1) \le 2$ and $\Phi_{q}^{\mathscr{A\!B}}$ is EA if
$q \le \frac{1}{4}$. The resolution (\ref{EA-resolution}) holds
true for all the above two-qubit EA channels, i.e.
Proposition~\ref{prop-2} reproduces the exact results (see
Fig.~\ref{figure4}a and Appendices~\ref{appendix-d} and
\ref{appendix-f}). As far as Corollary~\ref{corollary-biseparable}
is concerned, our analysis shows that it allows us to detect the
EA property of a smaller set of maps
$\Phi_{q_1}^{\mathscr{A}}\otimes\Phi_{q_2}^{\mathscr{B}}$ (see
Appendix~\ref{appendix-e} and Fig.~\ref{figure4}a). Analyzing
channels acting on qutrit-qubit systems, we succeeded in
constructing resolution (\ref{EA-resolution}) for a subset of EA
channels which is slightly smaller than the whole set of EA
channels (see Appendix~\ref{appendix-d} and Fig.~\ref{figure4}b
for local channels and Appendix~\ref{appendix-f} for global ones).

In what follows, we consider bipartite systems $\mathscr{A\!B}$
with $d^{\mathscr{A}}=d^{\mathscr{B}}=d$, where $d$ is arbitrary.

%%%%%%%%%%%%%%%%%%%%%%%%%%%%%%%%%%%%%%%%%%%%%%%%%%%%%%%%%%%%%%%%%%%
%%%%%%%%%%%%%%%%%%%%%%%%%%%%%%%%%%%%%%%%%%%%%%%%%%%%%%%%%%%%%%%%%%%
\begin{figure}
\includegraphics[width=8.5cm]{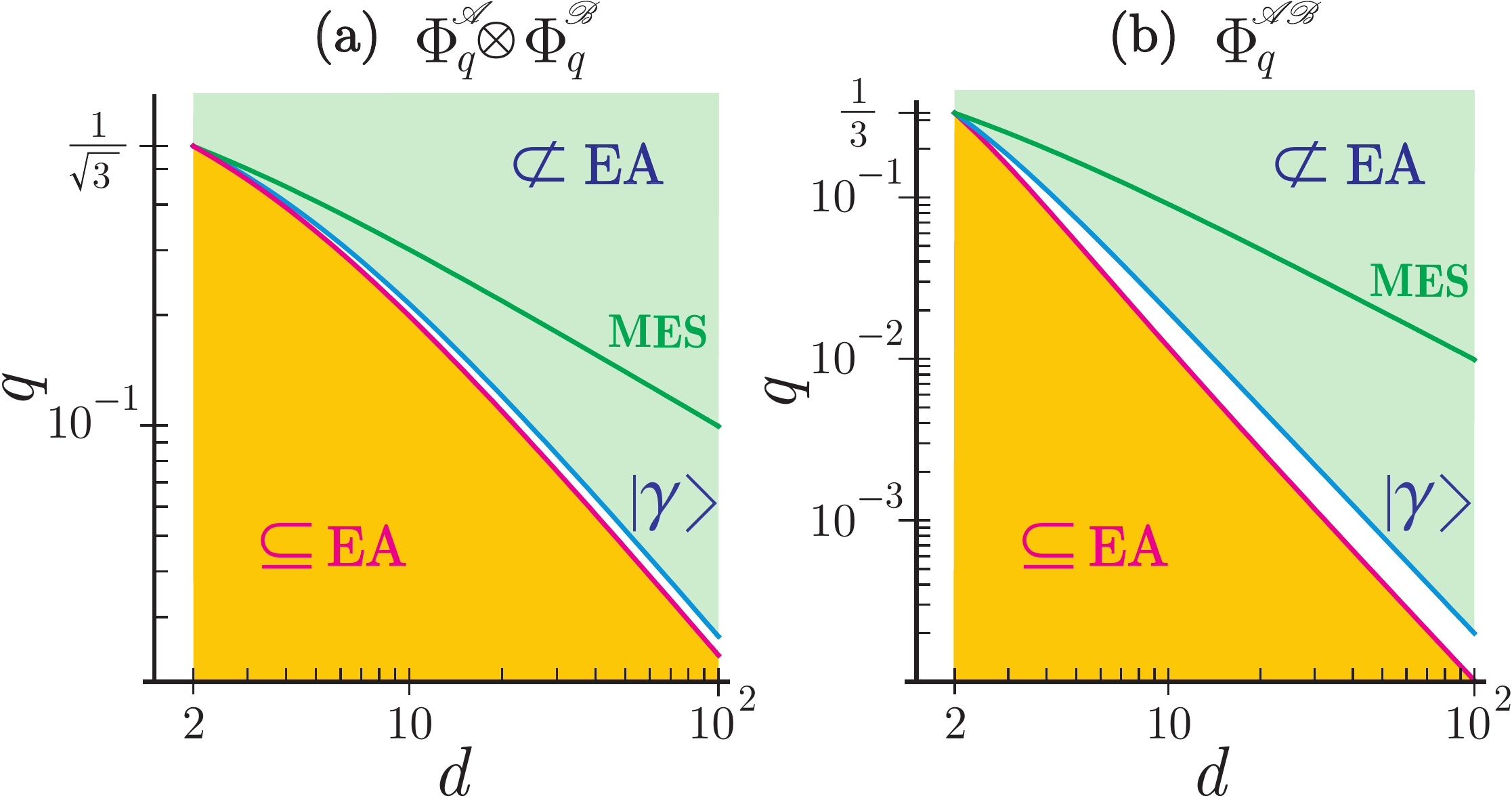}
\caption{\label{figure5} Local (a) and global (b) depolarizing
channels that surely annihilate or preserve entanglement of
$d\times d$ systems. Entanglement of the state $\ket{\gamma} =
\frac{1}{\sqrt{2}}(\ket{11}+\ket{dd})$ is more robust than that of
maximally entangled state $\ket{\Psi_{+}}$ (MES).}
\end{figure}
%%%%%%%%%%%%%%%%%%%%%%%%%%%%%%%%%%%%%%%%%%%%%%%%%%%%%%%%%%%%%%%%%%%
%%%%%%%%%%%%%%%%%%%%%%%%%%%%%%%%%%%%%%%%%%%%%%%%%%%%%%%%%%%%%%%%%%%

For a local channel
$\Phi_{q_1}^{\mathscr{A}}\otimes\Phi_{q_2}^{\mathscr{B}}$ it is
possible to find resolution (\ref{EA-resolution}) explicitly (see
Appendix~\ref{appendix-d}) if
\begin{equation}
\label{inequality} (d^2-1) \,q_1 q_2 \le 1 +
\frac{(d-2)(d+1)}{d+2}(q_1+q_2).
\end{equation}
Hence, for these values of parameters $q_1$ and $q_2$ the channel
is EA. Putting $q_{1,2}=q$ in (\ref{inequality}), we obtain that
$\Phi_q$ is 2LEA if $q \le q_{\rm EA}^{\rm local} =
\frac{d-2+d\sqrt{\frac{2d}{d+1}}}{(d-1)(d+2)}$ [see
Fig.~\ref{figure5}(a)], which determines a larger set than the EB
condition $q \le \frac{1}{d+1}$.

Consider a global depolarizing channel $\Phi_{q}^{\mathscr{A\!B}}$
acting on a pair of $d$-dimensional subsystems $\mathscr{A}$ and
$\mathscr{B}$ simultaneously. Such a noise is EB if and only if $q
\le \frac{1}{d^2+1}$. However, the noise disentangles
$\mathscr{A}$ from $\mathscr{B}$, hence it is EA, if $q \le q_{\rm
EA}^{\rm global} = (d+2)/[(d+1)(d^2-d+2)]$ (see
Fig.~\ref{figure5}b), which we showed by an explicit construction
of resolution~(\ref{EA-resolution}) in Appendix~\ref{appendix-f}.

Finally, we would like to give a counterintuitive example of an
entangled state which turns out to be more robust in the discussed
dissipative dynamics than the maximally entangled state
$\ket{\Psi_{+}^{\mathscr{A\!B}}} = \frac{1}{\sqrt{d}}
\sum_{i=1}^{d} \ket{i \otimes i}$. It can be readily seen that the
state $(\Phi_{q}^{\mathscr{A}}\otimes\Phi_{q}^{\mathscr{B}})
[\ket{\Psi_{+}^{\mathscr{A\!B}}} \bra{\Psi_{+}^{\mathscr{A\!B}}}]
\equiv (\Phi_{q^2}^{\mathscr{A}}\otimes{\rm
Id}^{\mathscr{B}})[\ket{\Psi_{+}^{\mathscr{A\!B}}}
\bra{\Psi_{+}^{\mathscr{A\!B}}}]$ becomes separable if
$\Phi_{q^2}^{\mathscr{A}}$ is EB, i.e. $q \le q_{\rm MES}^{\rm
local} = \frac{1}{\sqrt{d+1}}$. Similarly,
$\Phi_{q}^{\mathscr{A\!B}}[\ket{\Psi_{+}^{\mathscr{A\!B}}}
\bra{\Psi_{+}^{\mathscr{A\!B}}}] \equiv
(\Phi_{q}^{\mathscr{A}}\otimes{\rm
Id}^{\mathscr{B}})[\ket{\Psi_{+}^{\mathscr{A\!B}}}
\bra{\Psi_{+}^{\mathscr{A\!B}}}]$ becomes separable if
$\Phi_{q}^{\mathscr{A}}$ is EB, i.e. $q \le q_{\rm MES}^{\rm
global} = \frac{1}{d+1}$. Consider now a state
$\ket{\gamma^{\mathscr{A\!B}}} = \frac{1}{\sqrt{2}}(\ket{1 \otimes
1}+\ket{d \otimes d})$ which is not maximally entangled (if $d >
2$). Surprisingly, the states
$(\Phi_{q}^{\mathscr{A}}\otimes\Phi_{q}^{\mathscr{B}})[\ket{\gamma^{\mathscr{A\!B}}}
\bra{\gamma^{\mathscr{A\!B}}}]$ and
$\Phi_{q}^{\mathscr{A\!B}}[\ket{\gamma^{\mathscr{A\!B}}}
\bra{\gamma^{\mathscr{A\!B}}}]$ remain non-positive under partial
transposition~\cite{peres-1996,horodecki-1996} and, consequently,
are entangled if $q > q_{\rm nEA}^{\rm local} =
\frac{1+\sqrt{3}}{d+1+\sqrt{3}}$ and $q
> q_{\rm nEA}^{\rm global} = \frac{2}{d^2+2}$, respectively.
These results are depicted in Fig.~\ref{figure5}. A narrow gap
between channels that are surely EA and channels that are
definitely not EA underlines the importance of the state
$\ket{\gamma^{\mathscr{A\!B}}}$ in identifying potentially
dangerous noises in applications.

\section{Summary}
The analogy between the definitions of EB channels and CP maps
based on the consideration of map extensions, stimulates us to
introduce the concept of PEA maps as counterparts of positive maps
acting on a composite system $\mathscr{AB}$. The cone of PEA maps
is invariant under concatenation with partially positive maps.
This fact enabled us to find the necessary and sufficient
conditions for PEA maps as well as to find explicit form of the
dual cone of maps $\Phi_{\rm d.c.}^{\mathscr{A\!B}}[X] = \sum_{k}
\tr{F_k X} \xi_{{\rm BP}~k}^{\mathscr{A|B}}$, $F_k \ge 0$. This
form resembles measure-and-prepare procedures (being EB) but
differs in the use of block-positive operators. Based on these
criteria and in analogy with the entanglement theory one may
introduce the concept of EA witnesses. Imposing the conditions of
CP and trace-preservation on PEA maps we formulated sufficient
criteria for EA channels possessing a clear physical
interpretation illustrated in Fig.~\ref{figure3}. The derived
criteria were used in the analysis of local and global
depolarizing channels, for which we identified maximum noise
levels; going beyond those levels leads to entanglement
annihilation.

\begin{acknowledgments}
This work was supported by EU integrated project 2010-248095
(Q-ESSENCE), COST Action MP1006, APVV-0646-10 (COQI) and VEGA
2/0127/11 (TEQUDE). S.N.F. acknowledges support from the National
Scholarship Programme of the Slovak Republic, the Dynasty
Foundation, the Russian Foundation for Basic Research under
project 12-02-31524-mol-a.
\end{acknowledgments}

\appendix

\section{\label{appendix-a}Matrix representation of maps}
A linear map $\Phi:\cT(\cH_{\rm in})\rightarrow\cT(\cH_{\rm out})$
can also be defined via the $d_{\rm out}^2 \times d_{\rm in}^2$
matrix $\cE_{\Phi}$ with entries $(\cE_{\Phi})_{jk} =
(\tr{o_j^{\dag} o_j} \tr{\iota_k^{\dag} \iota_k})^{-1}
\tr{o_j^{\dag} \Phi[\iota_k]}$, where $\{\iota_k\}_{k=0}^{d_{\rm
in}^2-1}$ and $\{o_j\}_{j=0}^{d_{\rm out}^2-1}$ are orthogonal
operator bases in $\cT(\cH_{\rm in})$ and $\cT(\cH_{\rm out})$,
respectively. As a basis, we use normalized generalized Pauli
(Gell-Mann) matrices $\{\gamma_j\}_{j=0}^{d^2-1}$ satisfying the
relations $\gamma_j^{\dag}=\gamma_j$, $\tr{\gamma_j\gamma_k} =
\delta_{jk}$, and $\gamma_0 = \frac{1}{d} I_d$. Using such a
basis, one can readily see that, in case $d_{\rm in,out} = d$, the
matrix representation of depolarizing channel $\Phi_q$ reads
$\cE_{\Phi_q} = {\rm diag}(1,q,\ldots,q)$.

In a matrix representation, a concatenation of maps corresponds to
a conventional matrix product: $\cE_{\Upsilon \circ \Phi} =
\cE_{\Phi} \cE_{\Upsilon}$. Also, $\cE_{\Phi \otimes \Upsilon} =
\cE_{\Phi} \otimes \cE_{\Upsilon}$. These properties are
especially pleasing for diagonal matrices (depolarizing maps).

\section{\label{appendix-b}EB depolarizing channels} Let us make a change of
variable $q = \left[ d(2\mu-1)-1 \right]/(d^2-1)$, then the Choi
matrix $\Omega_{q}$ of the depolarizing map $\Phi_q$ is equal to
the partially transposed Werner state $\varrho_{\mu}^{\Gamma}$,
where $\varrho_{\mu} = \mu \frac{2}{d(d+1)}P_+ + (1-\mu)
\frac{2}{d(d-1)}P_-$ is a convex combination of projectors onto
symmetric and antisymmetric subspaces of
$\cH_d\otimes\cH_d$~\cite{werner-1989}. The state $\varrho_{\mu}$
is known to be separable if and only if it is positive under
partial transposition, i.e. $\frac{1}{2} \le \mu \le
1$~\cite{werner-1989}. It means that $\Omega_{q}$ is separable
and, consequently, $\Phi_q$ is EB if $q\in[-\frac{1}{d^2-1},
\frac{1}{d+1}]$.

\section{\label{appendix-c}Positive bipartite maps}
Since positive maps on operators $\cT(\cH_d\otimes\cH_d)$ are
quite needed, we define a two-parametric map $\Lambda_{st}$ by the
following matrix representation:
\begin{equation}
\label{lambda} \Lambda_{st}={\rm diag}(1,\underset{d^2-1~{\rm
times}}{\underbrace{s,\ldots,s}};\underset{d^2-1~{\rm
times}}{\underbrace{s,\underset{d^2-1~{\rm
times}}{\underbrace{t,\ldots,t}};\cdots;s,\underset{d^2-1~{\rm
times}}{\underbrace{t,\ldots,t}}}}).
\end{equation}

\noindent The map (\ref{lambda}) is surely positive if
\begin{equation}
\label{lamda-positivity} 0 \le s \le t \le
\frac{1}{d-1}+\left(1-\frac{1}{d-1}\right)s,
\end{equation}

\noindent which is validated by checking block-positivity of its
Choi matrix $\Omega_{\Lambda_{st}}^{\mathscr{A\!BA'\!B'}}$ via the
method of Ref.~\cite{skowronek-zyczkowski} [positivity of
operators
$\bra{y^{\mathscr{A'\!B'}}}\Omega_{\Lambda_{st}}^{\mathscr{A\!BA'\!B'}}\ket{y^{\mathscr{A'\!B'}}}\in\cT(\cH^{\mathscr{A\!B}})$].

For systems $\mathscr{AB}$, where $d^{\mathscr{A}} \ne
d^{\mathscr{B}}$, one can use a straightforward modification of
(\ref{lambda}) with an appropriate number of terms. Such a map
will be positive if (\ref{lamda-positivity}) is fulfilled for
$d=\max(d^{\mathscr{A}},d^{\mathscr{B}})$.

\section{\label{appendix-d}Local depolarizing EA channels}
For $d\times d$ systems, the local depolarizing channel
$\Phi_{q_1}^{\mathscr{A}}\otimes\Phi_{q_2}^{\mathscr{B}}$ is
compatible with resolution (\ref{EA-resolution}) and,
consequently, is EA whenever $q_1$ and $q_2$ satisfy inequality
(\ref{inequality}). The resolution takes the form
\begin{equation}
\Phi_{q_1}^{\mathscr{A}}\otimes\Phi_{q_2}^{\mathscr{B}} = \mu
(\Phi_{p}^{\mathscr{A}}\otimes {\rm Id}^{\mathscr{B}}) \circ
\Lambda_{s_1 t_1}^{\mathscr{A\!B}} + (1-\mu)({\rm
Id}^{\mathscr{A}} \otimes \Phi_{p}^{\mathscr{B}}) \circ
\Lambda_{s_2 t_2}^{\mathscr{A\!B}}, \nonumber
\end{equation}

\noindent where
\begin{equation}
\mu=\frac{1}{2}+\frac{d+1}{2d}(q_2-q_1), \qquad
-\frac{1}{d^2-1}\le p \le \frac{1}{d+1}, \nonumber
\end{equation}

\noindent i.e. $\Phi_p$ is EB. Inequality (\ref{inequality})
transforms into equality if $p = \frac{1}{d+1}$. The maps
$\Lambda_{s_i t_i}$, $i=1,2$, are given by formula (\ref{lambda}),
where
\begin{eqnarray}
&& s_{1,2} =
\frac{2(d+1)}{d+2}\frac{(d+1)q_{2,1}-q_{1,2}}{d+(d+1)(q_{2,1}-q_{1,2})}, \nonumber\\
&& t_{1,2}=\frac{1}{d-1}+ \left(1-\frac{1}{d-1}\right)s_{1,2}.
\nonumber
\end{eqnarray}

For qutrit-qubit system, one should substitute the corresponding
EB maps $\Phi_{p_1 \le 1/4}^{\mathscr{A}}$ and $\Phi_{p_2 \le
1/3}^{\mathscr{B}}$ for $\Phi_{p}$. Numerical optimization over
parameters $\mu$, $s_{1,2}$, and $t_{1,2}$ results in the area of
parameters $(q_1,q_2)$ shown in Fig.~\ref{figure4}b.

\section{\label{appendix-e}Application of Corollary~\ref{corollary-biseparable} to
local depolarizing two-qubit channels}

In the case of two qubits, we now find parameters $q_1$ and $q_2$
such that the Choi matrix $\Omega_{\Phi_{q_1} \otimes
\Phi_{q_2}}^{\mathscr{A\!BA'\!B'}}$ can be represented as a convex
sum of density operators separable w.r.t. partitions
$\mathscr{A|BA'\!B'}$ and $\mathscr{B|AA'\!B'}$, i.e.
\begin{eqnarray}
\Omega_{\Phi_{q_1} \otimes \Phi_{q_2}}^{\mathscr{A\!BA'\!B'}} =
\frac{1}{k_{\rm max}} & \sum\limits_{k=1}^{k_{\rm max}} & \big[
\mu \ket{\psi_k^{\mathscr{A}}}\bra{\psi_k^{\mathscr{A}}} \otimes
\varrho_k^{\mathscr{BA'\!B'}} \nonumber\\
&& + (1-\mu) \ket{\psi_k^{\mathscr{B}}}\bra{\psi_k^{\mathscr{B}}}
\otimes \tilde{\varrho}_k^{\mathscr{AA'\!B'}} \big]. \nonumber
\end{eqnarray}

\noindent This resolution takes place if the operators
$\frac{1}{2}\ket{\psi_k}\bra{\psi_k}$ form a symmetric
informationally complete POVM ($k=1,\ldots,4$) or the vectors
$\{\ket{\psi_k}\}$ are elements of a full set of mutually unbiased
bases ($k=1,\ldots,6$) (see, e.g.,~\cite{heinosaari-ziman-2012}),
\begin{eqnarray}
\varrho_k^{\mathscr{BA'\!B'}} &=& \big(a
\ket{\psi_k^{\ast\mathscr{A'}}}\bra{\psi_k^{\ast\mathscr{A'}}} + b
\ket{{\psi_k}_{\perp}^{\ast\mathscr{A'}}}\bra{{\psi_k}_{\perp}^{\ast\mathscr{A'}}}
\big) \nonumber\\
&& \otimes
\ket{\Psi_+^{\mathscr{BB'}}}\bra{\Psi_+^{\mathscr{BB'}}} + c
I^{\mathscr{BA'\!B'}}, \nonumber\\
\tilde{\varrho}_k^{\mathscr{AA'\!B'}} &=& \big( a
\ket{\psi_k^{\ast\mathscr{B'}}}\bra{\psi_k^{\ast\mathscr{B'}}} + b
\ket{{\psi_k}_{\perp}^{\ast\mathscr{B'}}}\bra{{\psi_k}_{\perp}^{\ast\mathscr{B'}}}
\big) \nonumber\\
&& \otimes
\ket{\Psi_+^{\mathscr{AA'}}}\bra{\Psi_+^{\mathscr{AA'}}} + c
I^{\mathscr{AA'\!B'}}, \nonumber
\end{eqnarray}

\noindent $\mu = (1-q_1)q_2/(q_1 + q_2 - 2 q_1 q_2)$, $a =
\frac{1}{2}(q_1 + q_2 + 8 q_1 q_2)$, $b = \frac{1}{2}(q_1 + q_2 -
4 q_1 q_2)$, and $c = \frac{1}{8}(1 - q_1 - q_2 + q_1 q_2)$.
However, the operators $\varrho_k$ and $\tilde{\varrho}_k$ are
positive only if $c \ge 0, a+c \ge 0, b+c \ge 0$. These
restrictions specify the region of parameters $q_1$ and $q_2$ by
the inequality $1 + 3 (q_1 + q_2) - 15 q_1 q_2 \ge 0$ which is
depicted in Fig.~\ref{figure4}a.

\section{\label{appendix-f}Global depolarizing EA channels}
For $d\times d$ systems, the global depolarizing channel
$\Phi_q^{\mathscr{A\!B}}$ is compatible with resolution
(\ref{EA-resolution}) and, consequently, is EA whenever $q \le
(d+2)/[(d+1)(d^2-d+2)]$. The resolution takes the form:
\begin{equation}
\Phi_{q}^{\mathscr{A\!B}} = \tfrac{1}{2} \big(
\Phi_p^{\mathscr{A}}\otimes{\rm Id}^{\mathscr{B}} + {\rm
Id}^{\mathscr{A}}\otimes\Phi_p^{\mathscr{B}} \big) \circ
\Lambda_{st}^{\mathscr{A\!B}}, \nonumber
\end{equation}

\noindent where $-\frac{1}{d^2-1} \le p \le \frac{1}{d+1}$ (i.e.
$\Phi_p$ is EB) and $\Lambda_{st}$ is given by formula
(\ref{lambda}) with $s=2/(d^2-d+2)$, $t=(d+2)s$.

For qutrit-qubit system, the weight factors
$(\frac{1}{2},\frac{1}{2})$ should be replaced by $(\mu,1-\mu)$, a
single positive map $\Lambda_{st}$ should be split into two
($\Lambda_{s_1t_1}$ and $\Lambda_{s_2t_2}$), the maps $\Phi_{p}$
should be replaced by the corresponding EB maps $\Phi_{p_1 \le
1/4}^{\mathscr{A}}$ and $\Phi_{p_2 \le 1/3}^{\mathscr{B}}$.
Numerical optimization over parameters $\mu$, $s_{1,2}$, and
$t_{1,2}$ shows that $\Phi_q^{\mathscr{A\!B}}$ is surely EA if $q
\le 0.21$ which is slightly less than the exact value
$\frac{1}{4}$.

\end{document}